\documentclass{article}

\usepackage[utf8]{inputenc}

\usepackage{amssymb,amsmath,amsthm,fullpage,color}

\usepackage{authblk}

\usepackage{mathpartir}

\usepackage{stmaryrd}

\usepackage{enumitem}

\usepackage{tikz-cd}

\usepackage{enguemacros}
\commentfalse

\theoremstyle{definition}
\newtheorem{definition}{Definition}
\newtheorem{lemma}[definition]{Lemma}

\newtheorem{corollary}[definition]{Corollary}

\newtheorem{theorem}[definition]{Theorem}
\theoremstyle{remark}
\newtheorem{remark}[definition]{Remark}

\title{
	On Up-to Context Techniques in the $\pi$-calculus
}

\author[1]{Enguerrand Prebet}
\affil[1]{Université de Lyon, ENS de Lyon, UCB Lyon 1, CNRS, INRIA, LIP}
\date{}

\begin{document}
	\maketitle
	
	\begin{abstract}
		We present a variant of the theory of compatible functions on relations,
		due to Sangiorgi and Pous. We show that the up-to context proof technique
		for bisimulation is compatible in this setting for two subsets of the
		pi-calculus: the asynchronous pi-calculus and a pi-calculus with
		immediately available names.
	\end{abstract}	
		
	Proving that two elements are bisimilar is usually done by relying on a
	relation that is a bisimulation and also contains the corresponding pair
	of elements.
	Up-to techniques provide a powerful way of simplifying such proofs, by
	requiring that a relation is only included in a bisimulation. One such example
	is the up-to context technique which allows us to remove contexts when playing
	along the bisimulation game.
	A general theory of those techniques is developed in 
	\cite{DBLP:conf/lics/Pous16},
        by focusing on
	the class of compatible functions on relations, 
	that are both sound up-to
	techniques and have nice compositional properties.
	
	In the $\pi$-calculus, up-to context is not a sound technique. In fact,
	bisimilarity is not even a congruence, due to the input prefix creating
	substitutions. However, in subcalculi like the Asynchronous $\pi$-calculus
	(\Api),
	bisimilarity is closed by substitution making it a congruence.
	Thus, the question of soundness of the up-to context technique
	for this subcalculus arises again.
	It is known 
	that up-to substitution
	is not compatible, and not even below the greatest compatible function
	(called the \emph{companion} in \cite{DBLP:conf/lics/Pous16}). 
	Thus it seemed that even if up-to substitution is sound, it could not
	be used in conjunction with other techniques without having to redo the 
	proofs all over again.
	
	Intuitively, the reason	why compatibility fails for up-to substitution is that
	compatibility assumes the knowledge about one step of transitions in the 
	bisimulation game, 
	while in 
	the proof of the congruence for \Api, the substitution is
        dealt with by looking
	at two successive transitions to deduce the behaviour of the
	next step of the program after substitution. More precisely,
	we need to look at two visible transitions to reason about an
	internal step. There is thus a distinction to be made between visible
	and internal steps
	which leads us to define the usual bisimulation function as the intersection
	of the two bisimulation functions represented by the diagrams below,
	with $\alpha$ ranging over visible actions.
	
	$$
	\begin{tikzcd}[column sep=tiny]
		P\arrow[d,"\alpha"]  & \R & Q\arrow[d,"\alpha"]\\
		P' & \R & Q'
	\end{tikzcd}\hspace{4cm}
	\begin{tikzcd}[column sep=tiny]
		P\arrow[d,"\tau"]  & \R & Q\arrow[d,"\tau"]\\
		P' & \R & Q'
	\end{tikzcd}
	$$
	
	In this paper, we propose a new notion of compatibility for
	a bisimulation function defined as $f\cap g$ (above $f$ would
	be the visible actions while $g$ would be the silent ones). 
	The key idea is to impose
	a stronger condition on $f$ and a weaker condition on $g$ 
	while preserving soundness.
	This allows us to define a framework where standard up-to techniques,
	including full up-to context, are both compatible and sound.
	
	We show this result for two subcalculi where bisimilarity is a congruence
	and similarly for weak bisimilarity.
	
	We thank Damien Pous and Davide Sangiorgi for helpful discussions about
	this work.

\section{Compatibility and Soundness}
In this section, we present some standard results about compatibility and their
usage to show the soundness of up-to techniques (Section~\ref{s:previous}).
Then we introduce \emph{compatibility with a function} that is a generalisation
of compatibility (Section~\ref{s:comp:with}). This notion still provides a 
soundness result while keeping nice properties of compatible functions
(like being composable).
\subsection{Previous work}\label{s:previous}
	Here, we recall standard results for compatibility from
	\cite{DBLP:conf/lics/Pous16}.
	
	\begin{definition}[Compatibility]
	$f$ is $g$-compatible if $f\circ g \subseteq g\circ f$
	\end{definition}

	\begin{definition}[Soundness]
		$f$ is $g$-sound via $f'$ if $f'$ is extensive and
		$\R \subseteq (g\circ f)(\R)$ implies $f'(\R) \subseteq b(f'(\R))$
	\end{definition}
	
	Compatible functions can be composed freely in a modular fashion.
	\begin{lemma}
		If $f_1,f_2$ are $g$-compatible, then $f_1 \circ f_2$ is $g$-compatible.
		
		If $g$ is monotone, we also have that $f_1 \cup f_2$ is $g$-compatible.
	\end{lemma}
	
	Compatible functions are useful as they are sound up-to techniques.
	\begin{lemma}\label{l:comp:snd}
		If $f$ is monotone and $g$-compatible, then $f$ is $g$-sound via $f^\omega$.
	\end{lemma}
	
	However, there are sound up-to techniques that are not exactly compatible.
	We can recover some of them using compatibility up-to.
	\begin{definition}[Compatible up-to]\label{d:comp:upto}
		$f$ is $g$-compatible up to $f'$ when $f'$ is expansive and
		$f \circ g \subseteq g \circ f' \circ f$.
	\end{definition}
	
	Compatible functions up to $f'$ can be related to compatible functions when $f'$
	is also compatible ensuring the soundness of such functions.
	\begin{lemma}
		If $f'$ is idempotent, monotone and expansive, $g$-compatible and $f$ is 
		$g$-compatible up to $f'$, then $f' \circ f$ is $g$-compatible.
	\end{lemma}
	\begin{proof}
		\begin{align*}
			f \circ g &\subseteq g \circ f' \circ f\\
			f' \circ f \circ g &\subseteq f' \circ g \circ f' \circ f & f'\text{ is monotone}\\
			&\subseteq g \circ f' \circ f' \circ f & f'\text{ is $g$-compatible}\\
			&\subseteq g \circ f' \circ f & f'\text{ is idempotent}
		\end{align*}
	\end{proof}
	
	In fact, compatible functions are a subset of compatible functions up to $f'$
	for any expansive $f'$.
	\begin{remark}
		If $f'$ is expansive, $g$ is monotone and $f$ is $g$-compatible, 
		then $f$ is also $g$-compatible up to $f'$.
	\end{remark}
	
\subsection{Compatibility with a function}\label{s:comp:with}
	Unfortunately, substitution is not a compatible function not even up to some 
	compatible $f'$.
	To see why, we call $b$ the bisimulation function associated to bisimilarity 
	for the $\pi$-calculus (see Section~\ref{sec:upto:pi}). We need the
	following lemma where $\top$ is the universal relation:
	\begin{lemma}\label{l:damien}
		Taking notations from CCS, we have $(a.\outC{c} | \outC{c}, a|\outC{c})
		\in b^2(\top)$
		but\\
		$((a.\outC{c} | \outC{c})\sub{a}{c}, 
		(a|\outC{c})\sub{a}{c})\notin b^2(\top)$
	\end{lemma}
	\begin{proof}
		First, by definition $(\outC{c},\nil)\in\top$ so $(a.\outC{c},a),
		(\outC{c}|\outC{c},\outC{c})\in b(\top)$. Therefore, 
		$(a.\outC{c} | \outC{c}, a|\outC{c}) \in b^2(\top)$.

		Then, we have $(a.\outC{a} | \outC{a})\sub{a}{c} = 
		a.\outC{a} | \outC{a} \stra{\tau} \outC{a}$ and the only transition
		that the second process can do is $(a | \outC{a})\sub{a}{c} = 
		a | \outC{a} \stra{\tau} \nil$.
		Thus, as $\outC{a} \stra{a}$ but $\nil \not\stra{a}$, we have 
		$(\outC{a}, \nil)\notin b(\top)$, meaning that 
		$(a.\outC{a} | \outC{a}, a | \outC{a})\notin b^2(\top)$.
	\end{proof}
	Being compatible up-to some compatible function implies
	being smaller than some other compatible function ($f \subseteq f'\circ f$ 
	in Lemma~\ref{d:comp:upto}). So it is enough to show that substitution 
	is not included in the companion $t$, 
	which is the greatest compatible function. As $t$ is
	compatible, we have that $t \circ b^2(\top) \subseteq b^2(\top)$.
	Thus, Lemma~\ref{l:damien} implies that $\texttt{sub}(b^2(\top)) \not\subseteq
	t(b^2(\top))$.
	This entails that up-to context is not compatible in the $\pi$-calculus.
	
	This example is asynchronous and as we will see later, up-to substitution
	is sound for \Api. Our goal is to adapt the notion of compatibility so that
	it captures up-to substitution. 
	
	The main idea is to proceed in two steps. We first focus on visible actions and
	establish compatibility with respect to the corresponding bisimulation function.
	Then, we exploit that result to prove a weaker version of compatibility on
	internal actions, which we define now.
	
	\begin{definition}[Compatibility with]
		We say that $f$ is $h$-compatible with $g$ (or $g,h$-compatible)
		if $f \circ (g \cap h) \subseteq h \circ f$.

	\end{definition}
	Intuitively, with the knowledge we have about $g$, we are able to prove a sort of
	compatibility result on $h$.
	
	The main use of ``compatibility with'' is to prove $g\cap h$-soundness,
	in which case we do not need to prove exactly $g,h$-compatibility. 
	For instance $g$-compatibility and $g^2,h$-compatibility
	is sufficient (see Thereom~\ref{th:comp:with:snd}). 
	
	In Section~\ref{sec:usage}, we exploit this approach taking $g$ as the
	bisimulation function restricted to visible actions and $h$ the one restricted
	to internal actions.

	``Compatibility with'' can be linked back to standard compatibility as follows:
	\begin{lemma}\label{l:backto:comp}
		If $f$ is $g\cap h$-compatible, then $f$ is $g,h$-compatible and $h,g$-compatible.
		
		If $f$ is monotone, $g$-compatible and $g,h$-compatible, then 
		$f$ is $g\cap h$-compatible.
	\end{lemma}
\begin{proof}~
	\begin{itemize}
		\item $f \circ (g\cap h) \subseteq (g\cap h) \circ f \subseteq g \circ f$ (and $h \circ f$ respectively)
		\item The following proof uses Lemma~\ref{l:mon:cap}, presented below.
		\begin{align*}
			f \circ (g \cap h) &= f \circ (g \cap (g \cap h))\\
			&\subseteq f \circ g \cap f \circ (g \cap h) & \text{by Lemma~\ref{l:mon:cap}}\\
			&\subseteq g \circ f \cap f \circ (g \cap h) & f \text{ is }
			g\text{-compatible}\\
			&\subseteq g \circ f \cap h \circ f & f \text{ is } g,h\text{-compatible}\\
			&= (g \cap h)\circ f
		\end{align*}
		
	\end{itemize}
\end{proof}

	If $f$ is a compatible function for $g\cap h$, like it is the case in
	Section~\ref{sec:usage}, we cannot show directly that $f$ is both $g$-compatible
	and $g,h$-compatible. Indeed, we only have the latter by Lemma~\ref{l:backto:comp}.
	We show in Section~\ref{sec:usage} that in the case of the $\pi$-calculus, things
	can be handled smoothly.
		
	\medskip
	
	Since we often use monotone functions and intersections, we rely on the
	following lemma:
	\begin{lemma}\label{l:mon:cap}
		If $f$ is monotone, then for any set $A,B$, $f(A\cap B) \subseteq f(A)\cap f(B)$.
		Similarly, for any function $g,h$, $f \circ (g \cap h) \subseteq f \circ h \cap f \circ g$.
		
		Conversely, for any $g,h$, $f \circ g \cup f \circ h \subseteq f \circ (g\cup h)$.
	\end{lemma}

	We can also build the composition and union of $g,h$-compatible functions under 
	mild assumptions.
	
	\begin{lemma}\label{l:bi:compose}
		If $f_1,f_2$ are $g$-compatible, $g,h$-compatible and monontone,
		then $f_1 \circ f_2$ is $g,h$-compatible.
	\end{lemma}
	\begin{proof}
		By Lemma~\ref{l:backto:comp}, $f_1$ and $f_2$ are $g\cap h$-compatible, so $f_1
		\circ f_2$ is too, meaning that $f_1 \circ f_2$ is $g,h$-compatible.
	\end{proof}
	
	\begin{lemma}
		If $f_1,f_2,h$ are monotone and $f_1,f_2$ are both $g,h$-compatible,
		then $f_1\cup f_2$ is $g,h$-compatible.
	\end{lemma}
	\begin{proof}
		\begin{align*}
			(f_1 \cup f_2) \circ (g \cap h) &= f_1 \circ (g \cap h)\cup 
			f_2 \circ (g \cap h)\\
			&\subseteq h \circ f_1 \cup h \circ f_2 & f_1,f_2\text{ are $g,h$-compatible}\\
			&\subseteq h \circ (f_1 \cup f_2)	& \text{by Lemma~\ref{l:mon:cap}}
		\end{align*}
	\end{proof}
	\begin{corollary}
		If $f,h$ are monotone, $g$-compatible and $g,h$-compatible,
		then $f^\omega$ is $g,h$-compatible.
	\end{corollary}
	
	Here, we state the theorem that is the equivalent of Lemma~\ref{l:comp:snd} for
	compatibility with. Intuitively,
	if $f$ is $g$-compatible, then we can use $g$ any number of times to show 
	the compatibility of $f$ with respect to $h$ (i.e $f$ is $g^m,h$-compatible), 
	and this is enough to prove that $f$ is a sound up-to technique for $g\cap h$.
	\begin{theorem}\label{th:comp:with:snd}
		If $f,g,h$ are monotone and $f$ is both $g$-compatible, $g^m,h$-compatible 
		(with $m \geq 1$) 
		then $f$ is $g\cap h$-sound via $f'^\omega$ with $f'=\bigcup_{i\leq m}f^i$.
	\end{theorem}
	\begin{proof}
		First, notice that $\R \subseteq ((g\cap h)\circ f)(\R)$ implies 
		$\R \subseteq (g\circ f)(\R)$ and $\R \subseteq (h\circ f)(\R)$.
				
		On one side, we then have $\R \subseteq (g \circ f)^m (\R)$, meaning
		$\R \subseteq (g^m \circ f^m)(\R) $ by compatibility and so
		$\R \subseteq (g^m \circ f')(\R)$.
		On the other side, we have that $\R \subseteq (h\circ f')(\R)$ as $h$ is monotone
		and $m \geq 1$.
		Therefore, $\R \subseteq ((g^m \cap h) \circ f')(\R)$.
		
		By Lemma~\ref{l:backto:comp}, $f$ is $g^m\cap h$-compatible, so
		$f'$ is too. Thus $f'$ is $g^m\cap h$-sound via $f'^\omega$ meaning
		$f'^\omega(\R) \subseteq (g^m\cap h)(f'^\omega(\R))$.
		Therefore we have $f'^\omega(\R) \subseteq h(f'^\omega(\R))$.
		As $f$ is $g$-compatible, we can also prove that $f'^\omega(\R) \subseteq g(f'^\omega(\R))$.
		
		By combining the two, we obtain that 
		$f'^\omega(\R) \subseteq (g\cap h)(f'^\omega(\R))$.

\end{proof}

To prove $g\cap h$-soundness, we are thus interested in showing $g$-compatibility,
and $g^m,h$-compatibility. By showing a weakening of compatibility with, we are able
to compose functions with different $m$.
\begin{lemma}
	If $f$ is monotone, is $g,h$-compatible, and $g' \subseteq g$, 
	then $f$ is $g',h$-compatible.
\end{lemma}
\begin{proof}
	As $f$ is monotone, $f \circ (g' \cap h) \subseteq f \circ (g \cap h)$.
	So $f \circ (g' \cap h) \subseteq h \circ f$.
\end{proof}
\begin{lemma}
	If $f_1,f_2$ are monotone, $g$-compatible, $f_1$ is $g^m,h$-compatible, 
	$f_2$ is $g^n,h$-compatible with $m \geq n$ and $g \subseteq id$, 
	then $f_1 \circ f_2$ and $f_2 \circ f_1$ are $g^{m},h$-compatible.
\end{lemma}
\begin{proof}
	As $g \subseteq id$, so is $g^{m-n}$ meaning $g^m \subseteq g^n$. 
	Thus, we have that $f_2$ is $g^m,h$-compatible.
	
	Also, $f_1,f_2$ being $g$-compatible, they are also $g^m$-compatible.
	
	Therefore, by Lemma~\ref{l:bi:compose}, $f_1 \circ f_2$ and $f_2 \circ f_1$
	are $g^m,h$-compatible.
\end{proof}

\medskip

Compatibility with can also be combined with compatibility up to 
(see Definition~\ref{d:comp:upto}):
\begin{definition}
	$f$ is $g,h$-compatible up to $f'$ when $f'$ is expansive and 
	$f \circ (g\cap h) \subseteq h \circ f' \circ f$.
\end{definition}

\begin{lemma}
	If $f'$ is idempotent, monotone, $g,h$-compatible and $f$ is $g$-compatible up to $f'$,
	$g,h$-compatible up to $f'$, then $f'\circ f$ is $g,h$-compatible.
\end{lemma}
\begin{proof}
	\begin{align*}
		f \circ (g\cap h) &= f \circ (g\cap g\cap h)\\
		&\subseteq f \circ g \cap f \circ (g \cap h)& \text{by Lemma~\ref{l:mon:cap}}\\
		&\subseteq g \circ f' \circ f \cap h \circ f' \circ f\\
		&= (g \cap h) \circ f' \circ f\\
		f' \circ f \circ (g\cap h) &\subseteq f' \circ (g\cap h) \circ f' \circ f 
		& f'\text{ is monotone}\\
		&\subseteq h \circ f' \circ f' \circ f & f'\text{ is $g,h$-compatible}\\
		&\subseteq h \circ f' \circ f & f'\text{ is idempotent}
	\end{align*}
\end{proof}

\begin{lemma}
	If $f'$ is idempotent, $g,h$-compatible up to $f''$ and $f$ is $g,h$-compatible 
	up to $f'$, then $f'\circ f$ is $g,h$-compatible up to $f''$.
\end{lemma}
\begin{proof}
		\begin{align*}
		f \circ (g\cap h) &= f \circ (g\cap g\cap h)\\
		&\subseteq f \circ g \cap f \circ (g \cap h)& \text{by Lemma~\ref{l:mon:cap}}\\
		&\subseteq g \circ f' \circ f \cap h \circ f' \circ f\\
		&= (g \cap h) \circ f' \circ f\\
		f' \circ f \circ (g\cap h) &\subseteq f' \circ (g\cap h) \circ f' \circ f 
		& f'\text{ is monotone}\\
		&\subseteq h \circ f'' \circ f' \circ f' \circ f & f'\text{ is $g,h$-compatible up to }f''\\
		&\subseteq h \circ f'' \circ f' \circ f & f'\text{ is idempotent}
	\end{align*}
\end{proof}


\section{Up-to context in the $\pi$-calculus}
\label{sec:usage}
We apply the theory developed above to the $\pi$-calculus. We recall the syntax and operational semantics of the $\pi$-calculus in
Figure~\ref{f:sos} (symmetric transitions have omitted).

\begin{figure}[th]
	$$\begin{array}{rcl}
		P,Q& ::=& !G \OR P|Q \OR \res{a}P \OR G
		\\[.1em]
		G,G'& ::=& \nil \OR \out{a}{b}.P \OR \inp{a}{b}.P \OR \tau.P \OR G+G'\\\\
		\alpha&::=& \out{a}{b}\OR \bout{a}{b} \OR \einp{a}{b}\\
		\mu &::=& \alpha \OR \tau
	\end{array}$$
	\begin{mathpar}
		\inferrule[Inp]{ }{\inp{a}{b}.P\stra{\out{a}{c}}P\sub{c}{b}}\and 
		\inferrule[Out]{ }{\out{a}{b}.P\stra{\out{a}{b}} P} \and \inferrule[Rep]{!G|G\stra{\mu}G'}{!G\stra{\mu}G'}
		\and
		\inferrule[Open]{P\stra{\out{a}{b}}P'}{\res{b}P\stra{\bout{a}{b}}P'}
		\text{ if } b\neq a\and 
		\inferrule[Sum]{G\stra{\mu}G'}{G+G''\stra{\mu}G'}
		\and
		\inferrule[Res]{P\stra{\mu} P'}{(\new a)P\stra{\mu}(\new a)P}\text{ if }a\notin\mathrm{n}(\mu)
		\and
		\inferrule[Par]{P\stra{\mu}P'}{P\,|\,Q\stra{\mu}P'\,|\,Q}\text{ if }\mathrm{bn}(\mu)\cap \mathrm{fn}(Q)=\emptyset
		\\
		\inferrule[Comm]{P\stra{\inp{a}{b}}P'\and Q\stra{\out{a}{b}}Q'}{P|Q\stra{\tau} P'|Q'}
		\and
		\inferrule[Close]{P\stra{\inp{a}{b}}P'\and Q\stra{\bout{a}{b}}Q'}{P|Q\stra{\tau} \res{b}(P'|Q')}\text{ if } b \notin \mathrm{fn}(P)
	\end{mathpar}
	
	\caption{Syntax and Early Labelled Transition System of the full $\pi$-calculus}
	\label{f:sos}
\end{figure}

We now show the usage of compatibility with to prove the soundness of the
up-to context techniques in subcalculi of the $\pi$-calculus.
For that, we first show it using our framework on non-input contexts
(Section~\ref{sec:upto:pi}).
This result is not new, but along the way, we prove that these up-to context
techniques are $\overline{b}_{\alpha}$ which will be required to compose it later on.
Then, in Section~\ref{s:upto:inp}, we isolate the key property
(Definition~\ref{d:alias:prop}) that is needed
to have the soundness of up-to substitution, and show how it gives the 
soundness result for up-to context.
We end by giving two subcalculi, the Asynchronous $\pi$-calculus
and a $\pi$-calculus with immediately available names, where this property holds,
thus proving up-to context technique can be used for these calculi.

\subsection{Up-to behavioural relations and evaluation contexts}
\label{sec:upto:pi}

We briefly recall the simulation ($s,\dots$) and bisimulation functions
$b,\overline{b}$ and introduce their
weaker versions $b_{\alpha}, b_\tau$ that only impose conditions on visible and internal
actions respectively.

\begin{align*}
	s(\R) &\defi \set{(P,Q) \OR \text{ for all } \mu, P', P\stra{\mu}P'
		\text{ implies there exists } Q' \text{ s.t } Q\stra{\mu}Q', P' \RR Q'}\\
	s_{\alpha}(\R) &\defi \set{(P,Q) \OR \text{ for all } \alpha, P', P\stra{\alpha}P'
		\text{ implies there exists } Q' \text{ s.t } Q\stra{\alpha}Q', P' \RR Q'}\\
	s_{\tau}(\R) &\defi \set{(P,Q) \OR \text{ for all } P', P\stra{\tau}P'
		\text{ implies there exists } Q' \text{ s.t } Q\stra{\tau}Q', P' \RR Q'}\\
	b(\R) &\defi s(\R) \cap s(\R^{-1})\\
	b_{\alpha}(\R) &\defi s_{\alpha}(\R) \cap s_{\alpha}(\R^{-1})\\
	b_\tau(\R) &\defi s_\tau(\R) \cap s_\tau(\R^{-1})\\
	\overline{b} &\defi id\cap b\\
\overline{b}_{\alpha}(\R) &\defi id\cap b_{\alpha}\\
\overline{b}_{\tau}(\R) &\defi id\cap b_\tau
\end{align*}

The variant $\overline{b}(\R)$ only contains pairs that are also in $\R$, thus
corresponding to the notion of respectfulness. As all the up-to techniques we use
are monotone, any results for $b_{\alpha}$ (resp. $b_\tau$) also holds with their variant
$\overline{b}_{\alpha}$ (resp. $\overline{b}_\tau$).
We note $\sim$ for the bisimilarity.

\begin{remark}~
	\begin{itemize}
		\item $b = b_{\alpha} \cap b_{\tau}$, $\overline{b} = \overline{b}_{\alpha}\cap \overline{b}_\tau$.
		\item All functions are monotone.
		\item $\overline{b}, \overline{b}_{\alpha}, \overline{b}_\tau \subseteq id$
	\end{itemize}	
\end{remark}

We will now define some up-to techniques corresponding to evaluation contexts
and prove their $b_{\alpha}$-compatibility and $b_\tau$-compatibility with $b_{\alpha}$.
\begin{mathpar}
	\mathcal{F}_S(\R) = S \R S^{-1}
	\and
	\texttt{refl}(\R) = \set{(P,P)}
\end{mathpar}
\begin{lemma}~
$\mathcal{F}_\sim$, $id$, \texttt{refl}
		are $b_{\alpha}$-compatible.
\end{lemma}


Evaluation contexts contain parallel composition and restriction.
\begin{mathpar}
	\texttt{res}(\R) = \set{(\res{\many{a}}P, \res{\many{a}}Q) \OR P \RR Q}
	\and
	\texttt{pcomp}(\R) = \set{(P|P', Q|Q') \OR P \RR Q, P' \RR Q'}
\end{mathpar}

\begin{lemma}~
\texttt{res}, \texttt{pcomp} are $b_{\alpha}$-compatible.
\end{lemma}

By direct application of Lemma~\ref{l:backto:comp} and existing results of
$b$-compatibility \cite{DBLP:conf/concur/MadiotPS14}, we have that
$\mathcal{F}_\sim$, $id$, \texttt{refl},
\texttt{res} are $b_{\alpha},\bar{b}_\tau$-compatible,
and \texttt{pcomp} is $b_{\alpha},\bar{b}_\tau$-compatible up to \texttt{res}.

Thus, we are able to take the union and compose while remaining sound according
to Theorem~\ref{th:comp:with:snd}.

\begin{corollary}\label{c:upto:ev:ctxt}
	$(id \cup \mathcal{F}_\sim \cup \texttt{refl} \cup \texttt{res} 
	\cup \texttt{pcomp})^\omega$ is $b$-sound.
\end{corollary}

Corollary~\ref{c:upto:ev:ctxt} is not new, but we obtain it via 
$b_{\alpha}$-compatibility and $b_{\alpha},\bar{b}_\tau$-compatibility instead of $b$-compatibility.
This is used below to compose those techniques with up-to substitution which is
not $b$-compatible.

In fact, we can already go further and add the remaining non-input contexts.

\begin{mathpar}
	\texttt{tau}(\R) = \set{(\tau.P, \tau.Q) \OR P \RR Q}
	\and
	\texttt{out}(\R) = \set{(\out{a}{b}.P, \out{a}{b}.Q) \OR P \RR Q}
	\and
	\texttt{sum}(\R) = \set{(G_1+G_2, G_1'+G_2') \OR G_1 \RR G_1', G_2\RR G_2'}
	\and
	\texttt{rep}(\R) = \set{(!G, !G') \OR G \RR G'}
\end{mathpar}
\begin{lemma}~
	\begin{itemize}
		\item \texttt{tau}, is $b_{\alpha}$-compatible.
		\item $id\cup\texttt{out}$ is $\bar{b}_{\alpha}$-compatible.
		\item $id\cup \texttt{sum}$ is $b_{\alpha}$-compatible.
		\item $id\cup\texttt{rep}$ is $b_{\alpha}$-compatible up-to 
		$\texttt{pcomp}\cup id$
	\end{itemize}
\end{lemma}
%

Similarly, using Lemma~\ref{l:backto:comp} and existing results, we have that
$id\cup\texttt{tau}$ and $id\cup\texttt{out}$ are
$\overline{b}_{\alpha},\overline{b}_\tau$-compatible,
$id\cup\texttt{sum}$ is $\overline{b}_{\alpha},\overline{b}_\tau$-compatible
and $id\cup\texttt{rep}$ is $\overline{b}_{\alpha},\overline{b}_\tau$-compatible
up to $\texttt{pcomp}\cup id$.

\subsection{Up-to substitution and input for subcalculi of $\pi$}
\label{s:upto:inp}
Next up, we can add substitution and input related contexts. The substitution makes use
of Theorem~\ref{th:comp:with:snd} with $m > 1$.

The proof requires an additional property that is not true in general in the 
$\pi$-calculus.

\begin{definition}[Aliased Communication Property]\label{d:alias:prop}
  We say that a set of processes $\mathcal P$ satisfies the
    \emph{aliased communication property} if for all processes $P$ in
    $\mathcal P$, we have the following properties:
  
  	\begin{itemize}
	\item $P\stra{\out{a}{b}}\stra{\einp{c}{b}}P'$ implies
$P\sigma \stra{\tau} P'\sigma$ for all $\sigma$ s.t.\ $a\sigma=c\sigma$.
	\item $P\stra{\bout{a}{b}}\stra{\einp{c}{b}}P'$ implies
$P\sigma \stra{\tau} \res{b}P'\sigma$ for all $\sigma$ s.t.\
$a\sigma=c\sigma$.
  	\end{itemize}
\end{definition}

This property is for instance satisfied in the asynchronous $\pi$-calculus and used to
show that bisimilarity on asynchronous $\pi$-terms is closed by
substitution.

\begin{mathpar}
	\texttt{sub}(\R) = \set{(P\sigma, Q\sigma) \OR P \RR Q}
	\and
	\texttt{inp}(\R) = \set{(\inp{a}{b}.P, \inp{a}{b}.Q) \OR P \RR Q}
\end{mathpar}

\begin{lemma}~
	\begin{itemize}
		\item \texttt{sub} is $b_{\alpha}$-compatible
		and $b_{\alpha}^2,b_\tau$-compatible up to $\mathcal{F}_{\equiv}\circ\mathtt{res}$.
		\item \texttt{inp} is $\overline{b}_{\alpha},b_\tau$-compatible.
		\item $id \cup \texttt{inp}$ is $\overline{b}_{\alpha}$-compatible up to \texttt{sub}.
	\end{itemize}
\end{lemma}
\begin{proof}~
	\begin{itemize}
		\item
		We rely on \cite[Lemma 1.4.13]{SanWal}:
		\begin{lemma}~\label{l:dav}
		\begin{enumerate}
			\item If $P\sigma \stra{\alpha} P'$ then $P \stra{\alpha'}P''$ for
			some $\alpha',P''$ with $\alpha'\sigma = \alpha$ and $P''\sigma = P'$.
			\item If $P\sigma \stra{\tau} P'$ then
			\begin{enumerate}
				\item $P\stra\tau P''$ for some $P''$ with $P''\sigma = P'$, or
				\item $P\stra{\out{a}{b}}\stra{\einp{c}{b}}P''$ for some $P'',
				a,b,c$ with $a\sigma = c\sigma$ and $P''\sigma \equiv P'$, or
				\item $P\stra{\bout{a}{b}}\stra{\einp{c}{b}}P''$ for some $P'',
				a,b,c$ with $a\sigma = c\sigma$ and $\res{b}P''\sigma \equiv P'$.
			\end{enumerate}
		\end{enumerate}
		\end{lemma}
		$b_{\alpha}$-compatibility follows from the lemma and that if $Q\stra{\alpha'}
		Q'$, then $Q\sigma \stra{\alpha'\sigma} Q'\sigma$.
		
		Take a relation $R$ and some processes $P,Q$ such that $(P,Q) \in b^2_{\alpha}\cap b_\tau(\R)$. We want to show that for all $\sigma$,
		$(P\sigma,Q\sigma)\in 
		b_\tau\circ\mathcal{F}_{\equiv}\circ\mathtt{res}\circ\texttt{sub}(\R)$.
		
		Take $P\sigma \stra{\tau} P'$, by the lemma, we have three cases:
		\begin{enumerate}
			\item either $P\stra\tau P''$ for some $P''$ with $P''\sigma = P'$.
			Then, as $(P,Q)\in b_\tau(\R)$, $Q\stra{\tau}Q'$ and $P'\RR Q'$.
			Thus, $Q\sigma \stra{\tau} Q'\sigma$ and $(P'\sigma,Q'\sigma) 
			\in \texttt{sub}(\R) \subseteq 
			\mathcal{F}_{\equiv}\circ\mathtt{res}\circ\texttt{sub}(\R)$.
			
			\item or $P\stra{\out{a}{b}}\stra{\einp{c}{b}}P''$ for some $P'',
			a,b,c$ with $a\sigma = c\sigma$ and $P''\sigma \equiv P'$.
			Then, as $(P,Q)\in b^2_{\alpha}(\R)$, $Q\stra{\out{a}{b}}\stra{\einp{c}{b}}Q''$ and $P''\RR Q''$.
			Thus, by the Aliased Communication Property $Q\sigma \stra{\tau}
			Q''\sigma$ and $(P',Q''\sigma) \in 
			\mathcal{F}_{\equiv}\circ\texttt{sub}(\R) \subseteq
			\mathcal{F}_{\equiv}\circ\mathtt{res}\circ\texttt{sub}(\R)$.
			
			\item or $P\stra{\bout{a}{b}}\stra{\einp{c}{b}}P''$ for some $P'',
			a,b,c$ with $a\sigma = c\sigma$ and $\res{b}P''\sigma \equiv P'$.
			Then, as $(P,Q)\in b^2_{\alpha}(\R)$, $Q\stra{\bout{a}{b}}\stra{\einp{c}{b}}Q''$ and $P''\RR Q''$.
			Thus, by the Aliased Communication Property $Q\sigma \stra{\tau}
			\res{b}Q''\sigma$ and $(P',Q''\sigma) \in 
			\mathcal{F}_{\equiv}\circ\mathtt{res}\circ\texttt{sub}(\R)$.
		\end{enumerate}
		
		\item Trivial (no transition)

		\item We can prove $\texttt{inp} \subseteq b_{\alpha} \circ \texttt{sub}$.
		
		Then,
		\begin{align*}
			(id\cup \texttt{inp}) \circ \overline{b}_{\alpha} &= 
			\overline{b}_{\alpha} \cup \texttt{inp} \circ \overline{b}_{\alpha}\\
			&\subseteq {b}_{\alpha} \cup {b}_{\alpha} \circ \texttt{sub} \circ \overline{b}_{\alpha}\\
			&\subseteq {b}_{\alpha} \circ (id \cup \texttt{sub} \circ \overline{b}_{\alpha}) 
			& \text{by Lemma~\ref{l:mon:cap}}
		\end{align*}
		Thus we know that $id \cup \texttt{sub} \circ \overline{b}_{\alpha} \subseteq 
		\texttt{sub} \subseteq \texttt{sub} \circ (id\cup \texttt{inp})$.
		
		On the other hand, we have $(id\cup \texttt{inp}) \circ \overline{b}_{\alpha} 
		\subseteq (id\cup \texttt{inp}) \subseteq \texttt{sub} \circ 
		(id\cup \texttt{inp})$. So $(id\cup \texttt{inp}) \circ \overline{b}_{\alpha} 
		\subseteq {b}_{\alpha} \circ \texttt{sub} \circ (id\cup \texttt{inp}) \cap 
		\texttt{sub} \circ (id\cup \texttt{inp}) = \overline{b}_{\alpha} 
		\circ \texttt{sub} \circ (id\cup \texttt{inp})$.
	\end{itemize}
\end{proof}

The property defined in Definition~\ref{d:alias:prop} is only used to show that
\texttt{sub} is $b_{\alpha}^2,b_\tau$-compatible up to \texttt{res}. However, because the
compatibility of \texttt{inp} is shown up to \texttt{sub}, 
the soundness of the corresponding technique
relies on the compatibility result for \texttt{sub}.

Finally, if we aggregate all the results:
\begin{theorem}\label{t:full:ctxt}
	If the aliased communication property holds, then
	$(\mathcal{F}_{\sim}\cup id \cup \texttt{refl} \cup \texttt{sub}
	\cup\texttt{res}\cup \texttt{pcomp}\cup{\tt sum}\cup
	\texttt{rep} \cup \texttt{tau} \cup \texttt{out} \cup
	 \texttt{inp})^\omega$ is $\overline{b}$-sound.
\end{theorem}

\subsection{Subcalculi satisfying the aliased communication property}
We present two subcalculi satisfying the aliased communication property.
The property does not hold in general because of processes like $\outC{a}.b$.
Thus, we look at \Api, where outputs cannot guard processes, and processes
with immediately available names, where dually inputs cannot be guarded.

\paragraph{Asynchronous $\pi$-calculus.}
The asynchronous $\pi$-calculus is defined by imposing that outputs no longer
guard a process, meaning that there are forbidden in sums and
in $\out{a}{b}.P$, we have $P=\nil$.

\begin{lemma}
	\Api{} satisfies the Aliased Communication Property.
\end{lemma}
\begin{proof}
	This is the direct application of Lemma 5.3.2 (3) and (4) 
	in \cite{SanWal}.
\end{proof}
\medskip

\paragraph{Immediately available names.}
Immediately available names may only be used in input
as soon as the name is created.
This is a weaker notion than linear receptiveness or uniform receptiveness
\cite{DBLP:journals/tcs/Sangiorgi99} which impose
that exactly one input (resp. replicated input) must be accessible.

This discipline is formalised by the following typing rules where $\Gamma$
is the set of name that can be used as input.
\begin{mathpar}
	\inferrule{\typ{\emptyset}{P}}
	{\typ{\Gamma}{\tau.P, \out{a}{b}.P}}
	\and
	\inferrule{\typ{\Gamma}{G}}{\typ{\Gamma}{!G}}
	\and
	\inferrule{\typ{\emptyset}{P} \and a\in\Gamma}
	{\typ{\Gamma}{\inp{a}{b}.P}}
	\and
	\inferrule{\typ{\Gamma,a}{P}}{\typ{\Gamma}{\res{a}P}}
	\and
	\inferrule{\typ{\Gamma}P\and\typ{\Gamma}Q}{\typ{\Gamma}{P|Q}}
	\and
	\inferrule{\typ{\emptyset}G\and\typ{\emptyset}G'}{\typ{\emptyset}{G+G'}}
\end{mathpar}

Note that because of the typing rule for sum, inputs are forbidden in sums.

Typable processes form a subcalculus of the $\pi$-calculus.
Indeed, the set of typable processes is closed by transitions as expressed 
by the lemma below.
\begin{lemma}[Subject Reduction]
	If $\typ{\Gamma}{P}$ and $P\stra{\mu}P'$, then $\typ{\Gamma\cup\bn{\mu}}{P'}$.
\end{lemma}

\begin{lemma}
	The set of typable processes satisfies the Aliased Communication Property.
\end{lemma}

As a consequence of Theorem~\ref{t:full:ctxt}, the up-to context techniques
is sound for both subcalculi.

\subsection{The weak case}
We show how these results can be adapted to the weak case. The weak arrows are
defined as usual: $\wtra{}\defi \stra{\tau}^*$, $\wtra{\alpha} \defi 
\wtra{}\stra{\alpha}\wtra{}$, $\wtrah{\alpha} \defi \wtra{\alpha}$ and
$\wtrah{\tau} \defi  \wtra{}$.

We define the simulation functions for the weak case $ws,ws_\alpha,ws_\tau$,
the corresponding bisimulation functions $wb, wb_\alpha, wb_\tau$, and
their variant
$\overline{wb}, \overline{wb}_\alpha, \overline{wb}_\tau$ follow as expected.
\begin{align*}
	ws(\R) &\defi \set{(P,Q) \OR \text{ for all } \mu, P', P\stra{\mu}P'
		\text{ implies there exists } Q' \text{ s.t } Q\wtrah{\mu}Q', P' \RR Q'}\\
	ws_{\alpha}(\R) &\defi \set{(P,Q) \OR \text{ for all } \alpha, P', P\stra{\alpha}P'
		\text{ implies there exists } Q' \text{ s.t } Q\wtra{\alpha}Q', P' \RR Q'}\\
	ws_{\tau}(\R) &\defi \set{(P,Q) \OR \text{ for all } P', P\stra{\tau}P'
		\text{ implies there exists } Q' \text{ s.t } Q\wtra{}Q', P' \RR Q'}
\end{align*}
Weak bisimilarity is noted $\wba$.

Most results true in the strong case also hold in the weak case. We give details
about those whose statement or proof need to be adapted.

First, it is known that up-to weak bisimilarity is not a sound technique.
However, we can still use up-to strong bisimilarity but also use the
expansion preorder $\gtrsim$.

Take $s'(\R) \defi \set{(P,Q) \OR \text{ for all } \mu, P', P\stra{\mu}P'
	\text{ implies there exists } Q' \text{ s.t } Q\strah{\mu}Q', P' \RR Q'}$
where $\strah{\alpha} \defi \stra{\alpha}$ and $\strah{\tau} \defi \stra{\tau}^=$.
Then $\gtrsim$ is the largest relation $\R$ such that $\R \subseteq ws(\R)\cap 
s'(\R^{-1})$.

\begin{lemma}
	$\mathcal{F}_{\sim}, \mathcal{F}_{\gtrsim}$ is $wb_{\alpha}$-compatible. 
\end{lemma}

The aliased communication property needs also to be changed to use weak arrows
so that the proof of substitution goes without trouble.
\begin{definition}[Weak Aliased Communication Property]\label{d:alias:prop:w}
	We say that a set of processes $\mathcal P$ satisfies the
	\emph{weak aliased communication property} if for all process $P$ in
	$\mathcal P$, we have the following properties:
	
	\begin{itemize}
		\item $P\wtra{\out{a}{b}}\wtra{\einp{c}{b}}P'$ implies
		$P\sigma \wtra{} P'\sigma$ for all $\sigma$ s.t.\ $\sigma(a)=\sigma(c)$.
		\item $P\wtra{\res{b}\out{a}{b}}\wtra{\einp{c}{b}}P'$ implies
		$P\sigma \wtra{} \res{b}P'\sigma$ for all $\sigma$ s.t.\
		$\sigma(a)=\sigma(c)$.  		
	\end{itemize}
	
\end{definition}

Weak bisimilarity is not a congruence for sum. Indeed, we have
$\tau.a \wba a$ but $\tau.a+b \not\wba a+b$. Congruence is usually recovered by
considering \emph{non-degenerate} contexts, that is, contexts where the hole
is not directly under a sum operator. Therefore, we want to prove the soundness 
of up-to
non-degenerate contexts, and thus we use the up-to guarded sum technique
instead of the previous up-to sum technique:

$${\tt sum_g}(\R) = \set{(\sum_i G_i, \sum_i G'_i) \OR \forall i, (G_i,G'_i)\in
	(\texttt{tau}\cup\texttt{out}\cup\texttt{inp}\cup\texttt{refl})(\R)}$$

\begin{lemma}
	$id\cup{\tt sum_g}$ is $\overline{b}_{\alpha}$-compatible up to
	$\texttt{sub}\cup\texttt{refl}$ and 
	$\overline{b}_{\alpha},\overline{b}_\tau$-compatible up to $id\cup\texttt{refl}$.
\end{lemma}
\begin{proof}
	
We prove $\mathtt{sum_g} \subseteq b_\tau \circ (id \cup \texttt{refl})$
and $\mathtt{sum_g} \subseteq b_{\alpha} \circ (\texttt{sub} \cup \texttt{refl})$.

Suppose $\sum_i G_i \stra{\mu} G'$, then $G_{i_0} \stra{\mu}G'$.
\begin{itemize}
	\item If $(G_{i_0},G'_{i_0}) \in \texttt{refl}(\R)$, then 
	$\sum_i G'_i \stra{\mu} G'$ and $(G',G') \in \texttt{refl}(\R)$
	
	\item If $(G_{i_0},G'_{i_0}) \in \texttt{tau}(\R)$, then $\mu = \tau$,
	so $G'_{i_0} = \tau.G''$ with $G' \RR G''$ and
	$\sum_i G'_i \stra{\tau} G''$ and $(G',G'') \in id(\R)$
	
	\item If $(G_{i_0},G'_{i_0}) \in \texttt{out}(\R)$, then $\mu = \out{a}{b}$,
	$G_{i_0} = \out{a}{b}.G'$.
	So $G'_{i_0} = \out{a}{b}.G''$ with $G' \RR G''$ and
	$\sum_i G'_i \stra{\mu} G''$ and $(G',G'') \in \texttt{id}(\R)$.

	\item If $(G_{i_0},G'_{i_0}) \in \texttt{inp}(\R)$,
then $\mu = \einp{a}{c}$, $G_{i_0} = \inp{a}{b}.G''$ with $G' = G''\sub{c}{b}$.
So $G'_{i_0} = \inp{a}{b}.G'''$ with $G'' \RR G'''$ and
$\sum_i G'_i \stra{\mu} G'''\sub{c}{b}$ and $(G',G'''\sub{c}{b}) \in \texttt{sub}(\R)$.
\end{itemize}
\end{proof}

The proof for replication needs also to be changed. We show instead that
$id\cup\texttt{rep}$ is $\overline{b}_{\alpha},\overline{b}_\tau$-compatible
up to $\mathcal{F}_\sim \circ (\texttt{pcomp}\cup id)$. Intuitvely,
the problem is similar to the case of the sum, but because we have 
the law $!G \sim !G|G$, it does not break soundness.

The other proofs can be carried out without any modification, 
and we can then conclude
with the soundness of the whole up-to technique:
\begin{theorem}
	If the weak aliased communication property holds, then
	$(\mathcal{F}_{\gtrsim}\cup id \cup \texttt{refl} \cup \texttt{sub}
	\cup\texttt{res}\cup \texttt{pcomp}\cup\mathtt{sum_g}\cup
	\texttt{rep} \cup \texttt{tau} \cup \texttt{out} \cup
	\texttt{inp})^\omega$ is $\overline{wb}$-sound.
\end{theorem}

\medskip

We can now show that both subcalculi also satisfy the weak aliased communication
property.

In asynchronous $\pi$-calculus, we know that outputs, being asynchronous, 
may always be postponed as expressed below.
\begin{lemma}\label{l:key:async}
	If $P \stra{\out{a}{b}}\stra\mu P'$ then 
	$P \stra\mu\stra{\out{a}{b}} P'$.
	
	If $P \stra{\bout{a}{b}}\stra\mu P'$ and $b\notin\fn{\mu}$, then 
	$P \stra\mu\stra{\bout{a}{b}} P'$.
\end{lemma}
Thus, if $P\wtra{\out{a}{b}}\wtra{\einp{c}{b}}P'$, then
$P\wred\stra{\out{a}{b}}\stra{\einp{c}{b}}\wred P'$ so \Api{} satisfies the weak
aliased communication property.

\bigskip

For immediately available names, the reasoning is reversed. As inputs are immediately
available, we can show that they can be preponed:
\begin{lemma}\label{l:key:recep}
	If $P \stra\mu\stra{\einp{a}{b}} P'$ and $a\notin\bn{\mu}$, then
	$P \stra{\einp{a}{b}}\stra\mu P'$. 
\end{lemma}
Again, this lemma ensures that if $P\wtra{\out{p}{b}}\wtra{\einp{q}{b}}P'$, then
$P\wred\stra{\out{p}{b}}\stra{\einp{q}{b}}\wred P'$ and so we can conclude.

We can notice a symmetry between Lemmas~\ref{l:key:async} and \ref{l:key:recep},
the former delays outputs while the latter anticipates inputs.

This shows that both calculi are also a congruence for the weak bisimilarity
and that the up-to context technique is sound.

\begin{remark}
	Note that, in the (weak) aliased communication 
	property, we quantify over all names $a,c$, being the subject of
	the output and input respectively. If we impose 
	that the property holds for only some names, for instance if only 
	a subset of names are asynchronous or immediately available, then we have 
	the soundness of up-to substitution restricted to those
	asynchronous names (resp. immediately available names), 
	and up-to input that only carry
	asynchronous names (resp. immediately available names).
\end{remark}

\bibliographystyle{plain}
\bibliography{bib-func.bib}

\appendix
\newpage
\section{Language with lookahead}

Theorem~\ref{th:comp:with:snd} is sufficient to derive result with operator enabling after
more than transitions.

Consider the language
$$P ::= {\tt op}(P) \OR a.P \OR \nil$$
with the following semantic
\begin{mathpar}
	\inferrule{ }{a.P \stra{a} P}
	\and
	\inferrule{P \stra{a} \stra{a} P'}{{\tt op}(P)\stra{a} P'}
\end{mathpar}

It is known that up-to-bisimilarity-and-context is unsound.

However, one could tweak this language to make it sound.

\subsection{Two prefixes}
A first way would be to consider two prefixes $a$ and $b$, with $a.P \stra{a}P$ and
$b.P \stra{b}P$ but only $\inferrule{P \stra{b} \stra{b} P'}{{\tt op_1}(P)\stra{a} P'}$.

One could define simply $b_a,b_b$ two bisimulations using actions $a$ and $b$ respectively.

Then with ${\tt op_1}(\R) = \set{({\tt op_1}(P),{\tt op_1}(Q)) \OR P \RR Q}$, we have
${\tt op_1} \circ b_b \subseteq b_b \circ {\tt op_1}$ (there cannot be any $b$ transition)
and ${\tt op_1} \circ (b_b^2 \cap b_a) \subseteq b_a \circ {\tt op_1}$.

\subsection{Second version}

First, we need to show a variation of Theorem~\ref{th:comp:with:snd}. Here, we add the
condition that $g \subseteq id$ along with a rather technical condition $(\star)$ 
(stated to remain as general 
as possible), which is always verified in practice (for instance, if $f$ is expansive).
It allows us to show the soundness when we add an extra $h$ in front of $g^m$, i.e 
in the case where $f$ is $h\circ g^m, h$-compatible. 

\begin{theorem}\label{th:comp:with:snd:bis}
	If $f,g,h$ are monotone, $g \subseteq id$ and $f$ is both $g$-compatible, $h\circ g^m,h$-compatible 
	(with $m \geq 1$) and verifies the following:
	\[\forall n\in\mathbb{N}, f^n\circ(\bigcup_{i\leq mn+1}f^i) \subseteq \bigcup_{i\in\mathbb{N}}f^i \tag{$\star$}\]
	then $f$ is $g\cap h$-sound via $f^\omega$.
\end{theorem}
\begin{proof}
	Note that if $\R \subseteq ((g\cap h)\circ f)(\R)$ then $\R \subseteq (g\circ f)(\R)$
	and $\R \subseteq (h\circ f)(\R)$.
	
	We already know that $f$ is $g$-sound via $f^\omega$ so 
	$f^\omega(\R) \subseteq g(f^\omega(\R))$.

	As $g \subseteq id$, we have $h \circ g^m \cap h = h \circ g^m$. 
	Thus $f \circ h \circ g^m \subseteq h \circ f$.
	
	Now, let's show that the same inclusion holds for $h$.
	We prove by induction on $n$ that
	\[f^n \circ h \circ g^{mn} \subseteq h \circ f^n \tag{$\triangle$}\]
	
	For $n = 0$, trivial.
	For $n \geq 0$, 
	\begin{align*}
		f^{n+1} \circ (h \circ g^{m(n+1)})&=  f^n \circ f \circ h\circ g^m \circ g^{mn}\\
		&\subseteq f^n \circ h \circ f \circ g^{mn} &f\text{ is $g^m,h$-compatible}\\
		&\subseteq f^n \circ h \circ g^{mn} \circ f &f\text{ is $g$-compatible}\\
		&\subseteq h \circ f^n \circ f &\text{by induction}
	\end{align*}
	
	By monotonicity of $g$ and $f$, we can prove by a simple induction that
	$\R \subseteq (g\circ f)(\R)$ implies $\R \subseteq (g\circ f)^{mn}(\R)$.
	Then, we have $\R \subseteq (h\circ f)(\R) \subseteq (h \circ f)\circ (g\circ f)^{mn}(\R)$
	
	So $\R \subseteq (h \circ g^{mn} \circ f^{mn+1})(\R)$ by compatibility and
	if we note $f^\omega_{mn+1} = \bigcup_{i\leq mn+1}f^i$, as $g$ and $h$ are monotone
	we have	$\R \subseteq (h \circ g^{mn} \circ f^\omega_{mn+1})(\R)$.
	
	Thus:
	\begin{align*}
		\R &\subseteq (h \circ g^{mn}\circ f^\omega_{mn+1})(\R)\\
		f^n(\R) &\subseteq (f^n \circ h \circ g^{mn} \circ f^\omega_{mn+1})(\R)
		& f \text{ monotone}\\
		&\subseteq (h \circ f^n \circ f^\omega_{mn+1})(\R)
		&\text{ by }(\triangle)\\
		&\subseteq (h \circ f^\omega)(\R) & \text{by }(\star)
	\end{align*}
	
	Therefore, $f^\omega(\R) \subseteq h(f^\omega(\R))$
	
	In the end, we have $f^\omega(\R) \subseteq (g \cap h)(f^\omega(\R))$.
\end{proof}

This new theorem allows us to create a new operator ${\tt op_2}$ with
$\inferrule{P\stra{a}\stra{b} P'}{{\tt op_2}(P)\stra{a}P'}$ meaning it can now perform
action $a$ as its first action. In that case, ${\tt op_2} \circ (b_a \circ b_b \cap b_a)
\subseteq b_a \circ {\tt op_2}$.

We may also take $\inferrule{P\stra{\mu}\stra{b} P'}{{\tt op_2}(P)\stra{a}P'}$ 
with $\mu\in\set{a,b}$ and, noting $b$ for the whole bisimulation (i.e $b_a\cap b_b$),
show ${\tt op_2} \circ (b \circ b_b \cap b)
\subseteq b \circ {\tt op_2}$.

\subsection{Chaining further}
Compared to the previous examples where we split a bisimulation $b$ by splitting
the set of actions in two $b_\tau,b_{\alpha}$, one could also build incrementally smaller
bisimulation, for instance proving $b_{\alpha}^m, b$-compatibility instead of 
$b_{\alpha}^m,b_\tau$-compatibility.

This approach may require a bit more redundancy to prove compatibility results, 
but it does make statements easier to read. Here, we aim to decompose $b$ using
more than 2 functions, so we will use this incremental approach.

\begin{lemma}\label{l:chain:comp}
If $f$ is $g_1\cap g_2$-compatible and $g_2^n, g_3$-compatible, then
$f$ is $(g_1\cap g_2)^n \cap g_3$-compatible.
\end{lemma}
\begin{proof}
\begin{align*}
	f \circ ((g_1 \cap g_2)^n \cap g_3) &= f \circ ((g_1 \cap g_2)^n \cap g_2^n \cap g_3)\\
	&\subseteq f \circ (g_1 \cap g_2)^n \cap f \circ (g_2^n \cap g_3) 
	& \text{by Lemma~\ref{l:mon:cap}}\\
	&\subseteq (g_1 \cap g_2)^n \circ f \cap g_3 \circ f & \text{ by compatibility}\\
	&= ((g_1\cap g_2)^n \cap g_3) \circ f
\end{align*}
\end{proof}

This Lemma allows us to chain Lemma~\ref{l:backto:comp} into one compatible function.
For simplicity, we will assume $f$ is expansive.

\begin{theorem}\label{th:comp:with:snd:ter}
	If $f, g_i$ are monotone, $f$ is expansive, $g_1 \supseteq \dots \supseteq g_n$,
	and $f$ is $g_1$-compatible and for some $(m_i)_{i\leq n}$ with $m_i \geq 1$, $g_i^{m_i},g_{i+1}$-compatible,
	then $f$ is $g_n$-sound via $f^\omega$.
\end{theorem}
\begin{proof}
First, we define $h_i$ with $h_1 = g_1$ and $h_{i+1} = g_{i+1}\cap h_{i}^{m_i}$.
We show by induction on $i$ that $f$ is $h_i$-compatible.

When $i = 1$, this is true by assumption.

For $i = 2$, as $f$ is $h_1$-compatible, it is also $h_1^{m_1}$-compatible. Thus,
by Lemma~\ref{l:backto:comp}, $f$ is $g_2\cap h_1^{m_1}$-compatible.

For $i \geq 2$, as $f$ is $g_{i}\cap h_{i-1}^{m_{i-1}}$-compatible 
and $g_i^{m_i},g_{i+1}$-compatible, by Lemma~\ref{l:chain:comp},
$f$ is $g_{i+1}\cap (g_{i}\cap h_{i-1}^{m_{i-1}})^{m_i}$-compatible.

Next, we will show by induction that $\R \subseteq (h_i \circ f^{n_i})(\R)$ for some $n_i\geq 1$.
As $\R \subseteq (g_n \circ f)(\R)$, then $\R \subseteq (g_i \circ f)(\R)$ for all $i$,
so $\R \subseteq (h_1 \circ f)(\R)$.

Then if $\R \subseteq (h_i \circ f^{n_i})(\R)$, $\R \subseteq (h_i \circ f^{n_i})^{m_i}(\R)$,
and as $f$ is $h_i$-compatible, $\R \subseteq (h_i^{m_i} \circ f^{n_i*m_i})(\R)$.
Additionally, $\R \subseteq (g_{i+1} \circ f)(\R)$. As $f$ is expansive, 
$f \subseteq f^{n_i*m_i}$, so $\R \subseteq (g_{i+1} \circ f^{n_i*m_i})(\R)$.

Thus, $\R \subseteq ((g_{i+1}\cap h_i^{m_i}) \circ f^{n_i*m_i})(\R)$.

To sum up, $f$ is $h_n$-compatible and $\R \subseteq (h_n \circ f^{N})(\R)$ for
some $N \geq 1$. As $(f^N)^\omega = f^\omega$, we obtain that $f^\omega \subseteq g_n (f^\omega (\R))$.
\end{proof}

\subsection{Unverified approaches}

Having integers $n\geq 1$ as prefixes, i.e $n.P \stra{n} P$, and
\begin{mathpar}
	\inferrule{P\stra{n}\stra{n} P'}{{\tt op_3}(P)\stra{n+1}P'}
	\and
	\inferrule{P\stra{n}\stra{m} P'}{{\tt op_4}(P)\stra{n+m}P'}
\end{mathpar}

With $b_n$ the bisimulation obtained by looking at transitions $\stra{m}$ with $m\leq n$.

We should have ${\tt op_i} \circ b_1 \subseteq b_1 \circ {\tt op_i}$ for $i=3,4$,
and for all $n\geq 1$, ${\tt op_i} \circ (b_n^2 \cap b_{n+1}) \subseteq b_{n+1} \circ
{\tt op_i}$ for $i=3,4$.

Thus, for all $n\geq 1$, ${\tt op_3}$ and ${\tt op_4}$ are valid up-to techniques for $b_n$.
\engue{What about $b$?}

\begin{mathpar}
	\inferrule{P\stra{n}\stra{m} P'}{{\tt op_5}(P)\stra{n}P'}n>m
\end{mathpar}
We should have ${\tt op_5} \circ b_1 \subseteq b_1 \circ {\tt op_5}$,
and for all $n$, ${\tt op_5} \circ (b_{n+1}\circ b_n \cap b_{n+1}) \subseteq b_{n+1} \circ
{\tt op_5}$.

\engue{Maybe we can redo Theorem~\ref{th:comp:with:snd:ter} as with Theorem~\ref{th:comp:with:snd:bis} ?}

\end{document}